\documentclass[sigconf, nonacm]{acmart}

\newcommand\vldbyear{2024}
\newcommand\vldbworkshop{3rd International Workshop on Large-Scale Graph Data Analytics (LSGDA 2024)}
\newcommand\vldbauthors{\authors}
\newcommand\vldbtitle{\shorttitle} 
\newcommand\vldbavailabilityurl{https://tinyurl.com/mryfsuxp}
\newcommand\vldbpagestyle{plain} 

\usepackage{subfigure}
\usepackage[ruled,vlined,linesnumbered]{algorithm2e} 

\newtheorem{example}{\textit{\textbf{Example}}}
\newtheorem{theorem}{{\textbf{Theorem}}}

\newtheorem{lemma}{{\textbf{Lemma}}}

\newtheorem{definition}{{\textbf{Definition}}}

\newcommand{\myparagraph}[1]{\vspace{1mm} \noindent \textbf{#1}.}

\begin{document}
\title{Parallel Higher-order Truss Decomposition}

\author{Chen Chen}
\affiliation{%
  \institution{University of Wollongong}
  \city{Wollongong}
  \country{Australia}
}
\email{chenc@uow.edu.au}

\author{Jingya Qian}
\affiliation{%
  \institution{Zhejiang Gongshang University}
  \city{Hangzhou}
  \country{China}
}
\email{jingyaq.zjgsu@gmail.com}

\author{Hui Luo}
\affiliation{%
  \institution{University of Wollongong}
  \city{Wollongong}
  \country{Australia}
}
\email{huil@uow.edu.au}

\author{Yongye Li}
\affiliation{%
  \institution{Zhejiang Gongshang University}
  \city{Hangzhou}
  \country{China}
}
\email{yongyeli.zj@gmail.com}

\author{Xiaoyang Wang}
\affiliation{%
  \institution{University of New South Wales}
  \city{Sydney}
  \country{Australia}
}
\email{xiaoyang.wang1@unsw.edu.au}

\begin{abstract}
The $k$-truss model is one of the most important models in cohesive subgraph analysis. The $k$-truss decomposition problem is to compute the trussness of each edge in a given graph, and has been extensively studied. However, the conventional $k$-truss model is difficult to characterize the fine-grained hierarchical structures in networks due to the neglect of high order information. To overcome the limitation, the higher-order truss model is proposed  in the literature. However, the previous solutions only consider non-parallel scenarios. To fill the gap, in this paper, we conduct the first research to study the problem of parallel higher-order truss decomposition. Specifically, a parallel framework is first proposed. Moreover, several optimizations are further developed to accelerate the processing. Finally, experiments over 6 real-world networks are conducted to verify the performance of proposed methods. 
\end{abstract}

\maketitle

\pagestyle{\vldbpagestyle}
\begingroup\small\noindent\raggedright\textbf{VLDB Workshop Reference Format:}\\
\vldbauthors. \vldbtitle. VLDB \vldbyear\ Workshop: \vldbworkshop.\\ 
\endgroup
\begingroup
\renewcommand\thefootnote{}\footnote{\noindent
This work is licensed under the Creative Commons BY-NC-ND 4.0 International License. Visit \url{https://creativecommons.org/licenses/by-nc-nd/4.0/} to view a copy of this license. For any use beyond those covered by this license, obtain permission by emailing \href{mailto:info@vldb.org}{info@vldb.org}. Copyright is held by the owner/author(s). Publication rights licensed to the VLDB Endowment. \\
\raggedright Proceedings of the VLDB Endowment. 
ISSN 2150-8097. \\
}\addtocounter{footnote}{-1}\endgroup

\ifdefempty{\vldbavailabilityurl}{}{
\vspace{.3cm}
\begingroup\small\noindent\raggedright\textbf{VLDB Workshop Artifact Availability:}\\
The source code, data, and/or other artifacts have been made available at \url{\vldbavailabilityurl}.
\endgroup
}

\section{Introduction}

Graphs are widely used to model the complex relationships among entities, such as social networks, protein-protein interaction networks and finance networks~\cite{wang2016bring,cheng2021efficient,wu2024efficient}. As one of the most fundamental tasks in graph analysis, cohesive subgraph detection has been extensively studied, and different models are proposed in the literature, such as $k$-core, $k$-truss, $k$-plex and clique~\cite{wu2020maximum,chen2021edge,chen2021maximum}. 
Among them, $k$-truss has received a lot of attention for its excellent balance between density and efficiency. $k$-truss leverages the properties of triangle to model the strength of connections. 
Specifically, given a graph $G$ and positive integer $k$, the $k$-truss is the maximal subgraph $S$ of $G$, such that the support of each edge $e (u,v)$ in $S$ is no less than $k-2$.
The support of $e(u,v)$ is the number of triangles that contain the edge in $S$, i.e., the number of common neighborhoods of $u$ and $v$ in the subgraph.
The trussness of an edge is the largest $k$ that $k$-truss contains the edge, i.e., the edge is in $k$-truss but not $(k+1)$-truss. Then, the $k$-truss can be computed by returning all the edges with trussness no less than $k$.

However, the $k$-truss model also has certain limitations. 
$k$-truss more focuses on describing pairwise interactions in the subgraph, 
while real-world systems may have many higher-order interactions involving groups of three or more units, which limits its application. 
Therefore, the higher-order truss model is proposed and studied in the literature to better capture the hierarchy of substructures~\cite{chen2021higher}. In particular, given a graph $G$, an integer $k$, and a hop threshold $h$, the higher-order $(k,h)$-truss of $G$ is the maximal subgraph $S$, where the $h$-support of each edge in $S$ is no less than $k-2$. $h$-support is the number of $h$-hop common neighbors for a given edge. 
The following motivating example illustrates that higher-order truss enables us to find more find-grained substructures in the underlying graph.

\begin{figure*}
    \centering
	\subfigure[1-hop neighborhood]{
		\includegraphics[width=0.48\linewidth]{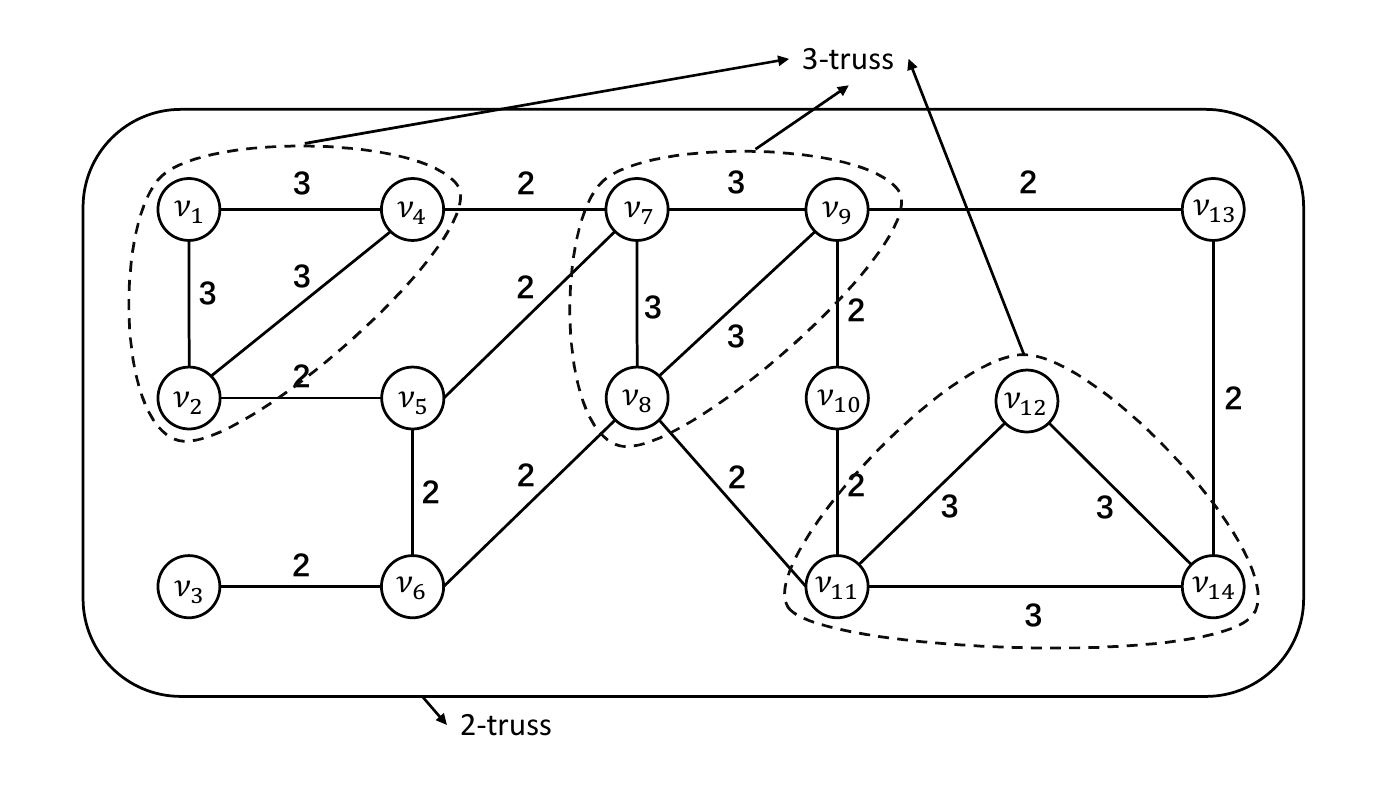}
        \label{fig1:a}
	}	
	\subfigure[$2$-hop neighborhood]{
		\includegraphics[width=0.48\linewidth]{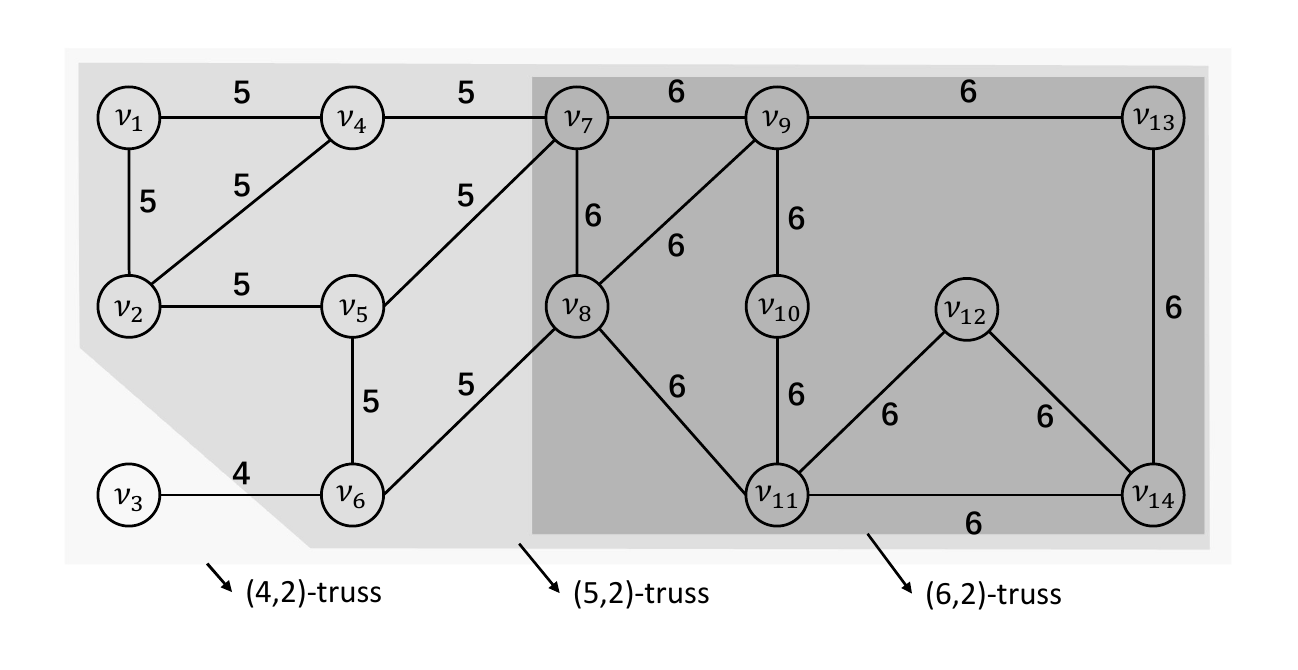}
		\label{fig1:b}
	}
    \caption{Motivation example of higher-order truss decomposition} \label{fig1}
\end{figure*}

\begin{example}
    Figure~\ref{fig1} shows a toy social network with 14 nodes. 
    The numbers on the edges in Figures~\ref{fig1}(a) and \ref{fig1}(b) denote the corresponding trussness and higher-order trussness, respectively. 
    If we adopt the $k$-truss model, the graph $G$ will be decomposed into $2$-truss and $3$-truss as shown in Figure \ref{fig1:a}
    which lacks a sense of hierarchy. 
    On the other hand, if we consider $h=2$, the $(k,h)$-truss model will find $(4,2)$-truss, $(5,2)$-truss and $(6,2)$-truss with a hierarchical structure. Specifically, the trusses are shown in Figure \ref{fig1:b} with different shades. 
\end{example}

For the higher-order scenario, given a graph $G$, the $h$-trussness of an edge $e$ is the largest $k$ such that the $(k,h)$-truss contains $e$, i.e., $e$ belongs to $(k,h)$-truss but not $(k+1,h)$-truss. The higher-order truss decomposition problem is to compute the $h$-trussness of all the edges for a given $h$. 
Similar as $k$-truss, based on $h$-trussness, we can easily retrieve the higher-order truss for any given $k$.
Naively, we can extend the peeling framework for $k$-truss decomposition to solve the higher-order case~\cite{chen2021higher}. However, it will generate much computation due to the larger search space in high order neighbors. Moreover, the computation cost would increase rapidly for higher $h$. To solve the problem, we conduct the first research to investigate the alternative of parallel solution. 

\myparagraph{Contributions} The contributions of this paper are summarized as follows.

\begin{itemize} 
    \item In this paper, we propose a parallel framework for higher-order truss decomposition based on the H-index computation paradigm. 
    
    \item Several optimization techniques are further proposed to accelerate the computation by considering the possibility of asynchronous update and redundant computation pruning. 
    
    \item We conducted extensive experiments on 6 real-world graphs to demonstrate the performance of proposed techniques. 
    
\end{itemize}

\myparagraph{Roadmap} The rest of the paper is organized as follows. 
We formally introduces the problem studied in Section 2, and present a baseline non-parallel approach in Section 3. 
In Section 4, we present the parallel framework  and optimization techniques proposed. 
In Section 5, we demonstrate the performance of the proposed techniques.
We introduce the related works in Section 6 and conclude the paper in Section 7.

\section{Preliminaries}

Given a graph $G=(V, E)$, we use $V$ and $E$ to represent the sets of vertices and edges in $G$, respectively. 
The set of neighbors of a vertex $v$ is denoted by $N_G(v)=\{u\in V|(u,v)\in E\}$.
The degree of $v$ is denoted by $deg_G(v)$, which is the number of its direct neighbors, i.e., $deg_G(v) = |N_G(v)|$.
Given a positive integer $h$, we use $N_G(v,h)$ to denote the set of $h$-hop neighbors of $v$, i.e., within distance $h$. 
Similarly, the $h$-degree of $v$, denoted by $deg_G(v,h)=|N_G(v,h)|$, refers to the number of its $h$-hop neighbors.
$S=(V_S,E_S)$ is a subgraph of $G$ with $V_S\subseteq V$ and $E_S=\{(u,v)|(u,v)\in E,u\in V_S,v\in V_S\}$.
Given an edge $e = (u, v)$, the support $sup_G(e)$ of $e$ is defined as the number of common neighbors of $u$ and $v$, i.e., the number of triangles containing $e$.
$sup_G(e)=|\bigtriangleup_G (e)|$, where $\bigtriangleup_G (e) = N_G(u) \cap N_G(v)$.



\begin{definition} [$k$-truss]
\label{def:k-truss}
    Given a graph $G$ and a positive integer $k$, a subgraph $S\subseteq G$ is the $k$-truss of $G$ if 
    $(i)$ $\forall e\in E_S$, $sup_S(e)\geq k-2$, and $(ii)$ $S$ is maximal, i.e., any supergraph of $S$ d.
\end{definition}

The trussness of an edge $e$, denoted as $t(e)$,
is the largest $k$ such that there is a $k$-truss that contains $e$, i.e., $e$ is in $k$-truss but not in $(k+1)$-truss. 
Given a graph $G$, the \textit{truss decomposition} problem is to compute the trussness of all the edges in $G$. 

\begin{definition} [common $h$-neighbor]
\label{def: common-h-neighbor}
    Given a graph $G$ and a distance threshold $h$, for an edge $e=(u,v)$ in $G$, if $w$ is an $h$-hops neighbor of both $u$ and $v$, then $w$ is a common $h$-neighbor of $e$. The set of common $h$-neighbors is denoted by $\bigtriangleup_G(e,h)$.
\end{definition}

\begin{definition} [$h$-support]
    Given a graph $G$ and a distance threshold $h$, for an edge $e$, the $h$-support of $e$, denoted by $sup_G(e,h)$, is the number of common $h$-neighbors of $e$, that is, $sup_G(e,h)=|\bigtriangleup_G (e,h)|$.
\end{definition}

\begin{definition} [$(k,h)$-truss]
    Given a graph $G$, a distance threshold $h$ and a positive integer $k$, a subgraph $S\subseteq G$ is the $(k,h)$-truss of $G$ if  
    $(i)$ $\forall e\in E_S$, $sup_S(e,h)\geq k-2$, and $(ii)$ $S$ is maximal, i.e., any supergraph of $S$ is not a $(k,h)$-truss.
\end{definition}

As we can see, when $h=1$, $(k,h)$-truss is the same as 
According to the above definitions, we define the $h$-trussness of edge $e$ to be the largest integer $k$ 
such that there is a $(k,h)$-truss containing $e$, denoted as $t(e,h)$. 
    
\myparagraph{Problem statement} Given a graph $G$ and a distance threshold $h$, in this paper, we aim to design efficient parallel solution for  
higher-order truss decomposition problem, i.e., compute the $h$-trussness of all the edges in $G$.

\section{Baseline Approach}
\label{sec:base}

In non-higher-order scenario, i.e., $h=1$, the truss model is containment. That is, $(k+1)$-truss is contained in $k$-truss.
Similarly, $(k,h)$-truss in the higher-order neighborhood also satisfies this property as shown in Lemma~\ref{le:contain}. The correctness can be easily verified based on the definition, and we omit the detailed proof here. 

\begin{lemma} \label{le:contain}
    Given a graph $G$, a distance threshold $h$ and a positive integer $k$, a $(k+1,h)$-truss of $G$ is a subgraph of $(k,h)$-truss.
\end{lemma}

According to Lemma~\ref{le:contain}, the $(k+1, h)$-truss can be obtained by peeling the $(k, h)$-truss of $G$. 
It means the edges peeled from $(k, h)$-truss are with $h$-trussness of $k$. Based on the peeling framework, we present the baseline non-parallel method by interactively removing the edges with the smallest $h$-support. To avoid the numerous sorting operations, we adopt bin sort as in core decomposition~\cite{batagelj2003m,sun2020stable}.

\begin{algorithm}[t]
        \caption{Baseline Solution} 
        \label{alg:base}
        \KwIn{graph $G(V,E)$, distance threshold $h$} 
        \KwOut{the $h$-trussness for each $e\in E$} 
        Initialization: $bin[ \cdot ]\leftarrow \{0\}$\;
        \ForEach{$e\in E$}
        {
            compute $h$-support $sup_G(e,h)$\;
            $bin[sup_G(e,h)+2]\leftarrow bin[sup_G(e,h)+2]\cup \{e\}$\;
        }
        ubtruss $\leftarrow max\{sup_G(e,h)|e\in E\}+2$\;
        \For{$k=2$ to ubtruss}
        {
            \While{$bin[k]\neq \emptyset$}
            {
                $e\leftarrow$ pick an edge from $bin[k]$\; 
                $t(e,h)\leftarrow k$\;
                $G\leftarrow G\setminus e$\;
                \ForEach{$e'\in \mathcal{E} _G(e,h)$ with $sup_G(e',h)>k-2$}
                {
                    compute and update $sup_G(e',h)$\;
                    move $e'$ to $bin[max(sup_G(e',h)+2,k)]$\;
                }
            }
        }
        $T\leftarrow \{t(e,h):e\in E\}$\;
        \textbf{return} $T$\;
\end{algorithm}

Algorithm~\ref{alg:base} presents the details of the baseline solution.
Specifically, we first initialize a set of bins in line 1, where $bin[i]$ is used to store all the edges with $h$-support equal to $i$.
Then, we compute $h$-support for all edges of graph $G$, and add these edges to the corresponding bins (lines 2-4).
We use ubtruss to denote the current maximum $h$-support+2.
Next, the bins are processed in increasing order of $k$ from $2$ to ubtruss. When the iteration reaches $k$, select an edge $e$ in $bin[k]$ to remove, and set its $h$-trussness to $k$, 
i.e., $t(e,h)=k$. When we remove $e$ from the graph, the $h$-support of the edge $e'\in \mathcal{E}_G(e,h)$ composed of its endpoint and its common $h$-neighbor vertices may decrease, 
so we update its $h$-support and move it to the corresponding bin (Lines 6-13). The iteration terminates when the $h$-trussness of all the edges are computed. 

\begin{figure}[t]
    \centering
	\includegraphics[width=\linewidth]{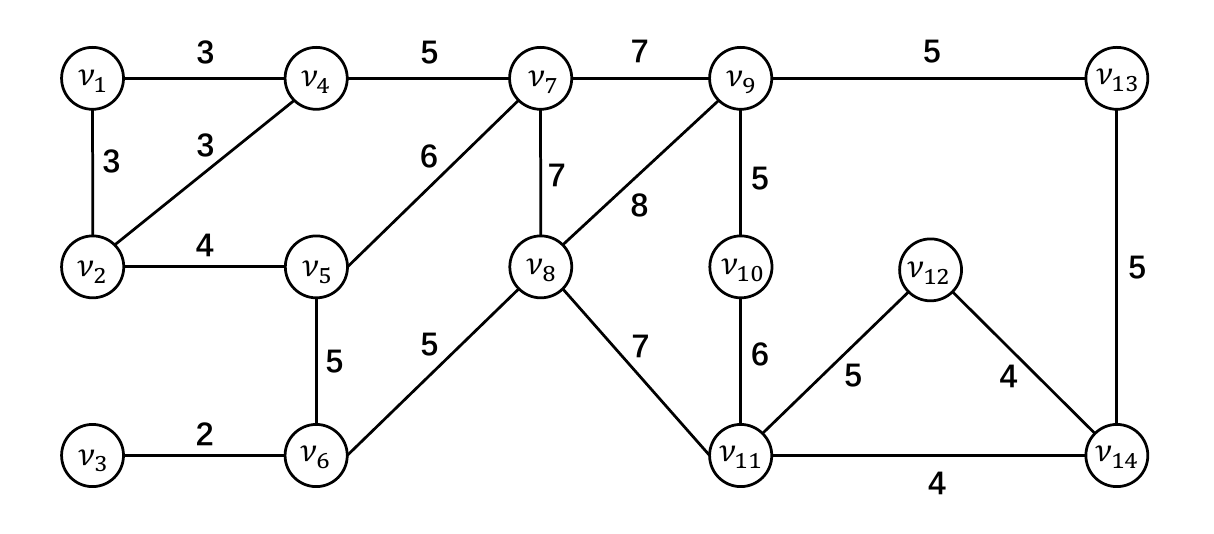}
    \caption{Initial 2-support for all edges} \label{fig:init2}
\end{figure}

\begin{example} 
Reconsider the graph in Figure~\ref{fig1}, Figure~\ref{fig:init2} shows the 
initialized 2-support values calculated for all the edges, i.e., the number on each edge.
Following the peeling approach, we start processing from the smallest value and find the edge $(v_3,v_6)$ with the current $2$-support of 2. 
We mark its $2$-trussness as 4 and delete it. The deletion of $(v_3,v_6)$ will influence the support of edges $(v_5,v_6)$ and $(v_6,v_8)$. We update their 2-supports and move the edges to the corresponding bins. Then, we repeat the above process until the graph is empty. 
\end{example}

\section{Our Solution}
\label{sec:our}

The baseline can effectively find the $h$-trussness for all edges. However, the non-parallel framework limits its scalability to large graphs. 
In this section, we first present some properties to facilitate the parallel computation of $h$-trussness (Section~\ref{subsec:hindex}). Then, we introduce the parallel computation approach developed (Section~\ref{subsec:parallel}) and some optimizations to accelerate the processing (Section~\ref{subsec:opts}).

\subsection{Convergence and Asynchrony in Higher-order Neighborhoods}
\label{subsec:hindex}

In \cite{sariyuce2017local}, the authors generalize the framework in \cite{lu2016h} to any kernel decomposition, including $k$-truss.
That is to say, for the $k$-truss model, we can use an H-index~\cite{hirsch2005index} computation paradigm to make it converge from the initial support to the trussness value.
In this paper, we use $\mathcal{H}(S)$ to denote the calculation of H-index for $S$. That is, $\mathcal{H}(S) = y$ denotes there are at least $y$ integers in $S$ that are no less than $y$. 
We briefly describe the idea as follows.

\begin{definition} [triangle edge pair]
    Given a graph $G$,
    a triangle edge pair of $e$ is the other two edges that can form a triangle with $e$. $\bigtriangleup_{ep}(e)$ denotes the set of all triangle edge pairs of $e$.
\end{definition}

Then, according to the H-index computation paradigm, for each edge $e$, we record all its triangle edge pairs. 
The initial $H_l^{(0)}sup(e)=sup_G(e)$ is used to denote 0-order H-index of edge $e$. For arbitrary integer $n$, the n-order H-index for edge $e$ is defined as 
$H_l^{(n)}sup(e)$ $=\mathcal{H}(\{min(H_l^{(n-1)}sup(e_i),$ $H_l^{(n-1)}sup(e_j))$ $|(e_i,e_j)\in \bigtriangleup_{ep}(e)\})$. 
The specific details are shown in the following lemma.

\begin{lemma}
    Given a graph $G$, for each edge $e\in G$, its H-index sequence $H_l^{(0)}sup(e),$ $H_l^{(1)}sup(e),$ $H_l^{(2)}sup(e)\dots$ 
    will converge to a value related to its trussness, i.e., $\lim_{n \to \infty} H_l^{(n)}sup(e)$ $=t(e)-2$.
\end{lemma}

Motivated by the above, in this paper, we further extend this paradigm to the higher-order scenarios. Suppose the key of each edge $e$ in the graph is related to its H-index, denoted as $H(e)$. The $h$-hops reachable path key is defined as follows.

\begin{definition} [$h$-hop reachable path key]
    Given a graph $G$, two vertices $v_i, v_j$ and a distance threshold $h$, 
    $P_{key}^{\leq h}(v_i,v_j)=max\{min\{H(e)|e\in p\}|p\in P\}$, where $p$ is a path from $v_i$ to $v_j$, distant $dist_G(v_i,v_j)$ $\leq h$, and $P$ is the set of paths. For simplicity, we  denote $P_{key}^{\leq h}(v_i,v_j)$ as $P(v_i,v_j)$.
\end{definition}

Based on the concepts introduced above, 
we can formally define the higher-order H-index for an edge $e$ in the graph.  
That is, $H^{(0)}sup(e)$ $=sup_G(e,h)$ and $H^{(n)}sup(e(u,v))$ $=\mathcal{H}(\{min(P^{(n)}(u,w),$ $P^{(n)}(v,w))|$ $w\in \bigtriangleup _G((u,v),h)\})$, 
where $P^{(n)}(v_i,v_j)$ $=max\{min\{H^{(n-1)}sup(e)|$ $e\in p\}|p\in P\}$.
Then, we prove the properties as follows. 

\begin{theorem}[Monotonicity]
\label{th:mono}
    Given a graph $G$, a threshold $h$ and a positive integer $n$, $H^{(n)}sup(e)\geq H^{(n+1)}sup(e)$ holds for each edge $e$ in $G$.
\end{theorem}

\begin{proof}
     Let edge $e=(u,v)$. We prove the theorem by mathematical induction.

    \begin{itemize}
      
    \item[$i)$] For $n=0$, $H^{(0)}sup(e)=sup_G(e,h)$, 
    and then $H^{(1)}sup(e)=\mathcal{H} (\{min(P^{(1)}(u,w),P^{(1)}(v,w))|w\in \bigtriangleup _G(e,h)\})$.
    From the property of $\mathcal{H}(\cdot )$, we can get $H^{(1)}sup(e)\leq |\bigtriangleup _G(e,h)|=sup_G(e,h)$.
    Thus, $H^{(0)}sup(e)\geq H^{(1)}sup(e)$ holds.

    \item[$ii)$] Assuming $n=m$, $H^{(m)}sup(e)\geq H^{(m+1)}sup(e)$ holds for all the edges. Then, we have
    $H^{(m+2)}sup(e)$ $=\mathcal{H}(\{min(P^{(m+2)}(u,w),P^{(m+2)}(v,w))|w\in \bigtriangleup _G(e,h)\})$. Based on the definition of $h$-hop reachable path key, we have 
    $P^{(m+2)}(v_i,v_j)=max\{min$ $\{H^{(m+1)}sup(e)|e\in p\}|p\in P\}$, 
    $P^{(m+1)}(v_i,v_j)=max$ $\{min\{H^{(m)}sup(e)|e\in p\}|p\in P\}$.
    Therefore, $P^{(m+1)}(v_i,v_j)$ $\geq P^{(m+2)}(v_i,v_j)$.

    \vspace{1mm}
    Thus, we have $H^{(m+1)}sup(e)=\mathcal{H}(\{min(P^{(m+1)}(u,w),$ $P^{(m+1)}(v,w))|w\in \bigtriangleup _G(e,h)\})\geq \mathcal{H}(\{min\{P^{(m+2)}(u,w),$ $P^{(m+2)}(v,w)\}|w\in \bigtriangleup _G(e,h)\})$, 
    That is, $H^{(m+1)}sup(e)\geq H^{(m+2)}sup(e)$ holds.
\end{itemize}
  
    Based on $i)$ and $ii)$, the theorem is correct.
\end{proof}

\begin{lemma}
    Given a graph $G=(V,E)$, a threshold $h$ and a positive integer $n$. Let $sup_{min}(G,h)=min\{sup_G(e,h)|e\in E\}$, 
    then $H^{(n)}sup(e)\geq sup_{min}(G,h)$ holds for each edge $e$ in $G$.
\end{lemma}

\begin{proof}
    We prove it by mathematical induction. Let edge $e=(u,v)$.
    First, $P^{(n)}(v_i,v_j)\geq min\{H^{(n-1)}sup(e)|e\in E\}$ can be obtained from the definition of $P^{(n)}(v_i,v_j)$.

    \begin{itemize}
    \item[$i)$] For $n=0$, $\forall e\in E$, $H^{(0)}sup(e)=sup_G(e,h)\geq sup_{min}(G,h)$ holds.

    \item[$ii)$] Assuming that $n=m$, $\forall e\in E$, $H^{(m)}sup(e)\geq sup_{min}(G,h)$ holds.
    We have $H^{(m+1)}sup(e)=\mathcal{H}(\{min(P^{(m+1)}(u,w),$ $P^{(m+1)}(v,w))|w\in \bigtriangleup _G(e,h)\})$, 
    $P^{(m+1)}(u,w)\geq min\{H^{(m)}sup(e)|e\in E\}\geq sup_{min}(G,h)$ and $P^{(m+1)}(v,m)$ is the same.
    So, here is $min(P^{(m+1)}(u,w),P^{(m+1)}(v,w))\geq sup_{min}(G,h)$.
    And $|\bigtriangleup _G(e,h)|=sup_G(e,h)\geq sup_{min}(G,h)$ 
    Thus, $H^{(m+1)}sup(e)\geq sup_{min}(G,h)$ holds.
    \end{itemize}

    Based on $i)$ and $ii)$, the lemma is correct.
\end{proof}

\begin{theorem}[Convergence]
\label{th:conver}
    Given a graph $G$, a threshold $h$ and a positive integer $n$,  we have $\lim_{n \to \infty} H^{(n)}sup(e)=t(e,h)-2$ holds for each edge $e$ in $G$.
\end{theorem}

\begin{proof}
   Suppose the edge $e=(u,v)$. We prove the theorem in two steps. 

    \begin{itemize}
    \item[$i)$] $H^{(\infty )}sup(e)\geq t(e,h)-2$.
    Let $g$ be $(t(e,h),h)$-truss containing $e$ in $G$, and $sup_{min}(g,h)=t(e,h)-2$ can be obtained from the property of the $k$-truss model, 
    so there is $H^{(n)}sup(e)(e\in G)\geq H^{(n)}sup(e)(e\in g)$.
    And by lemma 5 we can get $H^{(n)}sup(e)(e\in g)\geq sup_{min}(g,h)=t(e,h)-2$, so $H^{(n)}sup(e)(e\in G)\geq t(e,h)-2$ holds.

    \item[$ii)$] $t(e,h)\geq H^{(\infty )}sup(e)+2$.
    Let the edge set $E_s=\{e_s:e_s\in E\land H^{(\infty )}sup(e_s)\geq H^{(\infty )}sup(e)\}$, and let $G_s$ be the subgraph of $G$ induced by $E_s\cup e$.
    From definition $H^{(\infty )}sup(e)=\mathcal{H}(\{min(P^{(\infty )}(u,w),P^{(\infty )}(v,w))|w\in \bigtriangleup _G(e,h)\})$, 
    we can get $a)$ $|\bigtriangleup _G(e,h)|\geq H^{(\infty )}sup(e)$, 
    and $b)$ $\forall w\in \bigtriangleup _G(e,h)$, $min(P^{(\infty )}(u,w),P^{(\infty )}(v,w))\geq H^{(\infty )}sup(e)$ holds.
    From definition 18, we know $P^{(\infty )}(v_i,v_j)=max\{min\{H^{(\infty )}sup(e_l)|e_l\in p\}|p\in P\}$, 
    we get for $\forall w\in \bigtriangleup _G(e,h)$, there must be $p_i\in P$ in the path $p\in P$ from any endpoint ($u$ or $v$) of edge $e$ to $w$ 
    such that $min\{H^{(\infty )}sup(e_i)|e_i\in p_i\}\geq H^{(\infty )}sup(e)$.
    Let the set of all $p_i$ be $P_i$, then for $e_i\in P_i$ there is $H^{(\infty )}sup(e_i)\geq H^{(\infty )}sup(e)$.
    Let $E_i=\{e_i:e_i\in P_i\}$, and $E_i\subseteq E_s$ can be obtained by the definition of $E_s$, so there is $sup_{G_s}(e,h)\geq H^{(\infty )}sup(e)$.
    For $\forall e_s\in E_s$, $H^{(\infty )}sup(e_s)\geq H^{(\infty )}sup(e)$. Similarly, $sup_{G_s}(e_s,h)\geq H^{(\infty )}sup(e)$ can be obtained.
    Obviously, the above proves that $G_s$ is a $(H^{(\infty )}sup(e)+2,h)$-truss, and because of $e\in G_s$, the $h$-trussness of $e$ is $t(e,h)\geq H^{(\infty )}sup(e)+2$.
    \end{itemize}
    
    Based on $i)$ and $ii)$, the theorem is correct.
\end{proof}

\myparagraph{Asynchronous updating}
Based on the monotonic property in Theorem~\ref{th:mono}, we know that 
the $H^{(n)}sup(e)$ value obtained for each edge $e\in E$ in the H-index computation paradigm is non-increasing. Based on the convergence property in Theorem~\ref{th:conver}, for any positive integer $i$,  we have  $H^{(i)}sup(e)$  no less than $t(e,h)-2$.
More specifically, based on the definition of $(k,h)$-truss, we have $\mathcal{H}(\{min(P(u,w),P(v,w))|w\in \bigtriangleup _G(e,h)\})$ for each edge $e=(u,v)\in E$, 
where $P(a,b)=max\{min\{t(e_0,h)-2|e_0\in p\}|p\in P\}$.
When extending to $m$-order in the whole process, even if we use some of the calculated associated edge $e'$'s $H^{(m)}sup(e')$ directly in the calculation of $H^{(m)}sup(e)$ 
instead of using $H^{(m-1)}sup(e')$ strictly in order, it also satisfies monotonicity and convergence. 
So asynchronous updating can still guarantee a convergence to the $h$-trussness correlation value $t(e,h)-2$.

\begin{algorithm}[t]
        \caption{Parallel Decomposition Framework} 
        \label{alg:para}
        \KwIn{graph $G(V,E)$, distance threshold $h$} 
        \KwOut{the $h$-trussness for each $e\in E$}    
        \ForEach{$e\in E$ in parallel}
        {
            compute $h$-support $sup_G(e,h)$\;
            $H^{(0)}sup(e)\leftarrow sup_G(e,h)$\;
        }
        $flag\leftarrow true$\;
        $n\leftarrow 0$\;
        \While{$flag$} {
        $flag\leftarrow false$\;
        \ForEach{$e\in E$ in parallel}
        {
            
            $n\leftarrow n+1$\;
            $H^{(n)}sup(e)\leftarrow $ Algorithm~\ref{alg:calh}\;
            \If{$H^{(n)}sup(e)\neq H^{(n-1)}sup(e)$}
            {
                $flag\leftarrow true$\;
            }
        }
        }
        $t(e,h)\leftarrow H^{(n)}sup(e)+2$ for each $e\in E$\;
        $T\leftarrow \{t(e,h)|e\in E\}$\;
        \textbf{return} $T$\;
\end{algorithm}

\subsection{Parallel Algorithm}
\label{subsec:parallel}

Although the baseline algorithm can successfully achieves the decomposition of higher-order truss. However, it suffers from some limitations, 
such as low parallelism and high cost of edge deletion for $h$-support updates. 
To solve these problems, we introduce a parallel framework in this section.  
Some optimization techniques will be discussed in the next section.

According to Theorem~\ref{th:conver}, when $n$ of $H^{(n)}sup(e)$ is large enough, $H^{(n)}sup(e)$ can converge to $t(e, h)-2$ of $e$. 
Therefore, the general idea of the parallel decomposition framework is to 
is to iteratively compute $H^{(n)}sup(e)$ in parallel for each edge $e\in E$ until it converges.
In addition, based on the definition, we can easily obtain that when $H^{(n)}sup(e)=H^{(n-1)}sup(e)$ for all edges $e\in E$, it can be determined that $H^{(n)}sup(e)$ has converged. 
That is, for each $e\in E$ $H^{(n)}sup(e)=H^{(n-1)}sup(e)$ afterward, and by analogy, it can be obtained that with $i>0$, $H^{(n+i)}sup(e)$ will not change.
The parallel framework is shown in Algorithm~\ref{alg:para}.

\begin{algorithm}[t]
    \caption{Compute$H^{(n)}sup$} 
    \label{alg:calh}
    \KwIn{$H^{(n-1)}sup(e')$ for each $e'\in \mathcal{E} _G(e,h)$} 
    \KwOut{$H^{(n)}sup(e)$} 
    Initialization: $Q_1\leftarrow \emptyset $; visited$[\cdot]\leftarrow false$\;
    \ForEach{node $u\in e$}
    {
        \ForEach{$v\in N_G(u,1)$}
        {
           $P(u,v)\leftarrow H^{(n-1)}sup((u,v))$\; 
           $P'(u,v)\leftarrow H^{(n-1)}sup((u,v))$\;
           $Q_1.push(v)$; visited$[v]\leftarrow true$\;
        }
        $d\leftarrow 2$\;
        \While(){$d\leq h$}
        {
           $Q_2\leftarrow \emptyset $\;
           \While{$Q_1\neq \emptyset$}
           {
            $v\leftarrow Q_1.pop()$\;
            \ForEach{$w\in N_G(v,1)\land \bigtriangleup _G(e,h)$}
            {
                $new\_value$ $\leftarrow min\{H^{(n-1)}sup((v,w)),P'(u,v)\}$\;
                \If{visited $[w]$=false $\lor$ $new\_value$ $>P'(u,w)$}
                {
                    $P(u,w)\leftarrow new\_value$\;
                    $Q_2.push(u,w)$; visited$[w]\leftarrow true$\;
                }
            }
            \textbf{foreach} $w\in N_G(u,h)$  \textbf{do} $P'(u,w)\leftarrow P(u,w)$\;
           }
           $d\leftarrow d+1$; $Q_2\leftarrow Q_1$\;
        }
    }
    $H^{(n)}sup(e(u,v))\leftarrow \mathcal{H}(\{min(P^{(n)}(u,w),P^{(n)}(v,w))|w\in \bigtriangleup _G((u,v),h)\})$\;
    
\end{algorithm}

In Algorithm~\ref{alg:para}, we first calculate the $h$-support of all edges and set the support of 0-order H-index. That is, for each edge $e\in E$, its $H^{(0)}sup(e)$ is equal to its $h$-support (lines 1-3).
Then, we invoke Algorithm~\ref{alg:calh} for each edge $e$ in parallel to calculate the H-index of the support of the next order, and stop the iteration until $H^{(n)}sup(e)=H^{(n-1)}sup(e)$ for all edges (lines 6-12).
According to our previous theoretical analysis, the $h$-trussness of edge $e$ is equal to its $H^{(n)}sup(e)+2$. Finally, assign this value and return the the result (lines 13-15).

Algorithm~\ref{alg:calh} presents the details of computing the support of n-order H-index $H^{(n)}sup(e)$ for edge $e$.
We take $H^{(n-1)}sup(e')$ of the common h-neighbors $e'\in \mathcal{E} _G(e,h)$ of edge $e$ as input.
For each endpoint $u$ of $e$ to the $h$-hops reachable path key of its common $h$-neighbor, i.e., $P(u,w)$, 
we use a BFS method to tranverse from 1-hop to $h$ hops from $u$.
First, we compute the path key from $u$ to its 1-hop neighbor $v$ (lines 3-6), and then compute the path key to each neighbor $w$ of $v$.
We iterate this process until $d>h$ (lines 7-18).
We check $w\in \bigtriangleup _G(e,h)$ for each iteration, if it has not been visited, or there is a longer path that can generate a larger path key from $u$ to $w$, 
we assign $P(u,w)$ to the minimum value of $P'(u,v)$ of the previous path and $H^{(n)}sup((v,w))$ (lines 12-15).
Finally, the H-index computation of all $min(P^{(n)}(u,w),P^{(n)}(v,w))$ is performed to obtain the $H^{(n)}sup(e)$ of $e$ in line 19.

\begin{figure*}[t]
    \centering
    \includegraphics[width=0.9\linewidth]{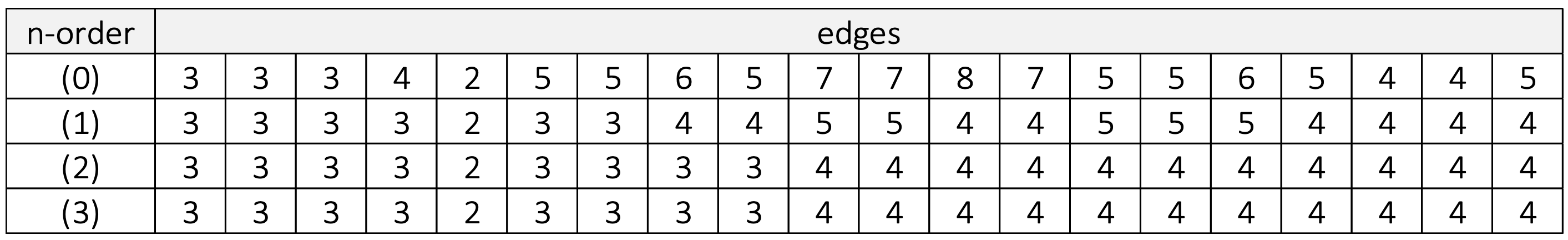}	
    \caption{Example of parallel iterative process} 
    \label{fig:paral}
\end{figure*}

\begin{example}
    Reconsider the example in Figure~\ref{fig:init2}. Figure~\ref{fig:paral}
    shows the updated support for all edges in each iteration. 
    $(i)$-order line represents the updated H-index in ith iterations. 
    $(0)$-order line represents the initial $2$-support of all edges.
    It shows that after four rounds of update iterations, 
    the $Hsup(e)$ of all edges are determined to converge to the point, i.e., equal to their 2-trussness.
\end{example}

\subsection{Optimizations}
\label{subsec:opts}

In this section, we propose some optimizations to further accelerate our proposed parallel algorithm.

\myparagraph{Asynchronous strategy}
In Algorithm~\ref{alg:calh}, each iteration is strictly in accordance with the provisions of using $(n-1)$-order to calculate $n$-order. 
However, in fact, according to the $n$-order H-index feature discussed above, we can update it asynchronously.
Before computing $H^{(n)}sup(e)$, the value of $H^{(n)}sup(e')$ may have already been computed for some $e'\in \mathcal{E} _G(e, h)$. 
Therefore, if we use $H^{(n)}sup(e')$ instead of $H^{(n-1)}sup(e')$ to calculate $H^{(n)}sup(e)$, it would accelerate the convergence
The correctness of this optimization is easy to be verified based on the analysis in Section~\ref{subsec:hindex}.

\myparagraph{Redundant computation pruning strategy}
There will be some redundant computation involved during the iterative computation of H-index. 
As a simple example, if there is an edge $e$, for each edge $e'\in \mathcal{E} _G (e,h)$ there is $H^{(n-2)}sup(e')=H^{(n-1)}sup(e')$.
Then, it is easy to get $H^{(n)}sup(e)=H^{(n-1)}sup(e)$. Thus, for the calculation of the edge $e$ of the $n$ order, 
the result can be directly obtained based on $(n-1)$ order.
    
In order to explore the decidable redundant computation in more detail, 
we study the effect of only one edge $e'\in \mathcal{E} _G (e,h)$ on the H-index computation result of edge $e$.
That is, except that $H^{(n-2)}sup(e')$ and $H^{(n-1)}sup(e')$ of one edge $e'\in \mathcal{E} _G (e,h)$ are unknown, 
and the other edges $e''\in \mathcal{E} _G (e,h)$ all satisfy $H^{(n-2)}sup(e'')=H^{(n-1)}sup(e'')$.

\begin{lemma}
    Given a graph $G$ and a distance threshold $h$, in the process of parallel truss decomposition, 
    the following cases $e'\in \mathcal{E} _G (e,h)$ will not trigger the change (i.e., decrease) of  $H^{(n)}sup(e)$.

    \begin{itemize}
        \item[$i)$] $H^{(n-2)}sup(e')=H^{(n-1)}sup(e')$;

    \item[$ii)$] $H^{(n-2)}sup(e')>H^{(n-1)}sup(e')$ and $H^{(n-2)}sup$ $(e')<H^{(n-1)}sup(e)$;

    \item[$iii)$] $H^{(n-2)}sup(e')>H^{(n-1)}sup(e')$ and $H^{(n-1)}$ $sup(e')\geq H^{(n-1)}sup(e)$.
    \end{itemize}
\end{lemma}
    
The proof is omitted and can be verified easily. 
Therefore, for each edge $e\in E$, if all its edges $e'\in \mathcal{E} _G (e,h)$ satisfy $i)$, $ii)$ or $iii)$ in Lemma 4, 
then we can directly get $H^{(n)}sup(e)=H^{(n-1)}sup(e)$ instead of computing $H^{(n)}sup(e)$.

\begin{table}[t]
\caption{Dataset Statistics}
\label{ta:data}
\begin{tabular}{cll}
\hline
Dataset                   & $|V|$     & $|E|$       \\
\hline
\hline
Yeast (YT)                & 1,870   & 2,227     \\
Human proteins Vidal (VL) & 3,133   & 6,726     \\
Sister cities (SC)        & 14,274  & 20,573    \\
Gnutella 30 (GA)          & 36,682  & 88,328    \\
Amazon TWEB 0302 (AM)     & 262,111 & 1,234,877 \\
Amazon MDS (AN)           & 334,863 & 925,872  \\
\hline
\end{tabular}
\end{table}

\section{Experiments}

In this section, we conduct extensive experiments on 6 real-world datasets to evaluate the performance of proposed techniques.

\subsection{Experiment setup}

\myparagraph{Algorithms} To evaluate the performance of proposed methods, we implement the follow algorithms in the experiment.

\begin{itemize}
    \item \textbf{Base.} The baseline algorithm proposed in Section~\ref{sec:base}.
    \item \textbf{Paral.} The parallel framework proposed in Section~\ref{subsec:parallel}.
    \item \textbf{Single.} The algorithm will perform a serial execution of parallel algorithm, i.e., Paral with single-threaded execution. It serves as a comparison algorithm that tests parallel acceleration ratio.
    \item \textbf{Asyn.} Paral equips with asynchronous update strategy proposed in Section~\ref{subsec:opts}. 
    \item \textbf{Paral+.} The parallel solution with all the optimization proposed in Section~\ref{subsec:opts}
\end{itemize}

\myparagraph{Dataset and workload} We employ 6 real-world graphs for experiment evaluation. The details are shown in Table~\ref{ta:data}. The datasets are public available on KONECT\footnote{http://konect.cc/networks/}. 
We implement all algorithms in C++ and all experiments are performed on a server with 2 Intel Xeon 2.10GHz CPUs and 256GB RAM running Linux. 
All algorithms are parallelized using OpenMP and by default use 20 threads for parallel computation. We mark the algorithm as \textbf{INF}, if it cannot be finished in 4 days. 
For each settings, we run 10 times and report the average value.

\begin{figure}[t]
    \centering
    \includegraphics[width=0.9\linewidth]{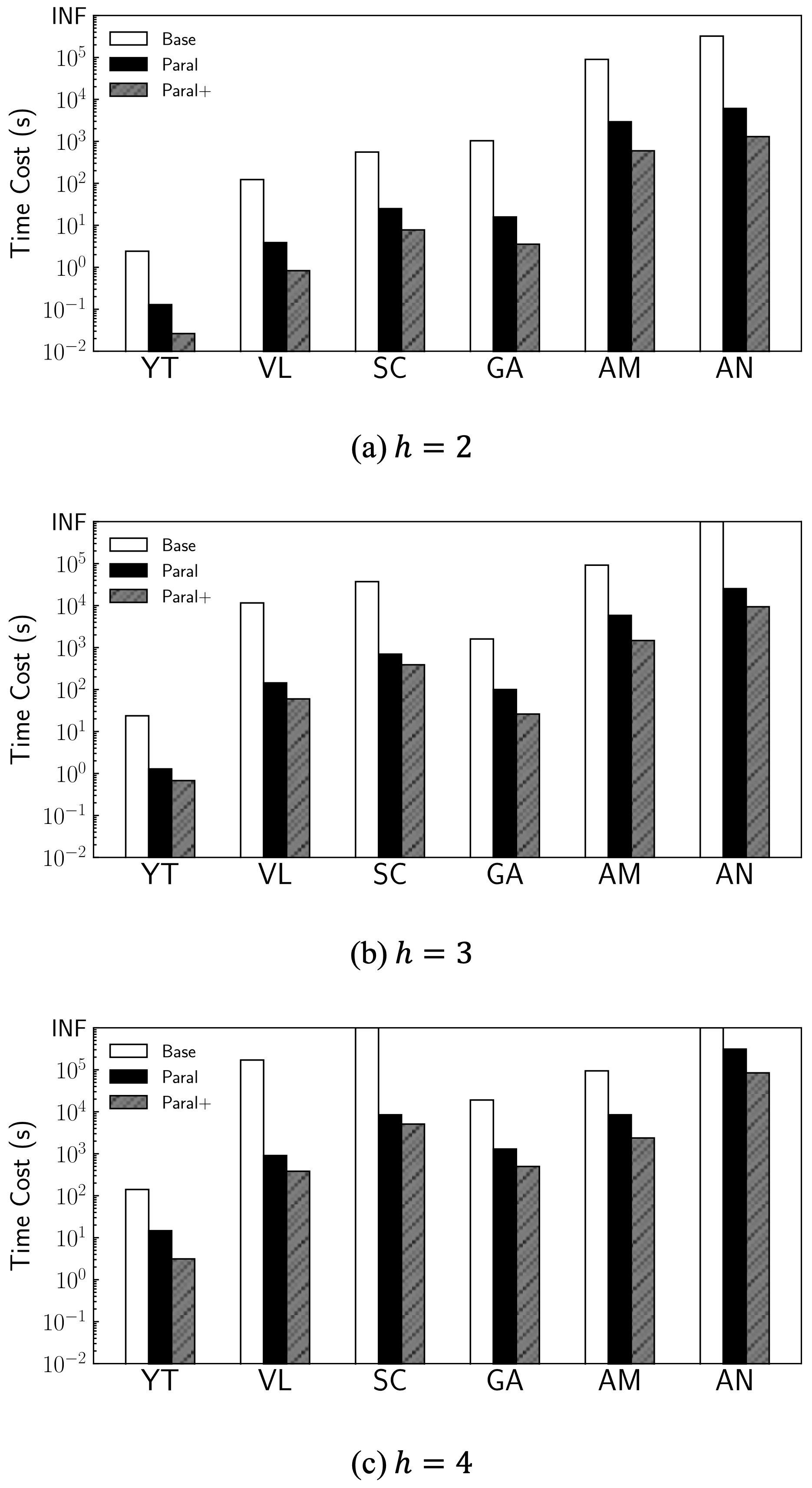}
    \caption{Efficiency evaluation on all the datasets}
    \label{fig:all_effi}
\end{figure}

\subsection{Evaluation Results}

\myparagraph{Efficiency evaluation on all datasets} Figure~\ref{fig:all_effi}, we report the response time of Base, Paral and Paral+ on all the dataset under different $h$ values. As we can see, parallel approaches are much faster than the baselines. Compared with Base, Paral+ can achieve up to 3 orders of magnitude speedup, demonstrating the advantage of parallel framework. Paral+ is faster than Paral, due to the further optimization techniques proposed in Section~\ref{subsec:opts}. Moreover, with the increase of $h$, more processing time is required by all the algorithms because of the larger search space involved. 

\myparagraph{Evaluation of parallel framework} To evaluate the impact of parallel framework, we report the speedup ratio by comparing Paral and Single. We conduct the experiments on 4 datasets, i.e., YT, VL, GA, and AM, where Base can finish in reasonable timeﬁ. The results are shown in Figure~\ref{fig:impact_para} by varying the number of threads. Note that, Single denotes the case with one thread. As we can observe, the speedup of parallel framework becomes larger with the increase of thread number. Paral is much more effective than single mode. For example, in GA, the speedup can achieve 32 with $h=3,4$. In some cases, the speedup is not significant. For example, in GA, the speedup is only 11 with 32 threads. This is because, many computations converge in first few rounds. The gain in rest computation is not significant by increasing the number of threads.

\begin{figure}[t]
    \centering
    \includegraphics[width=0.9\linewidth]{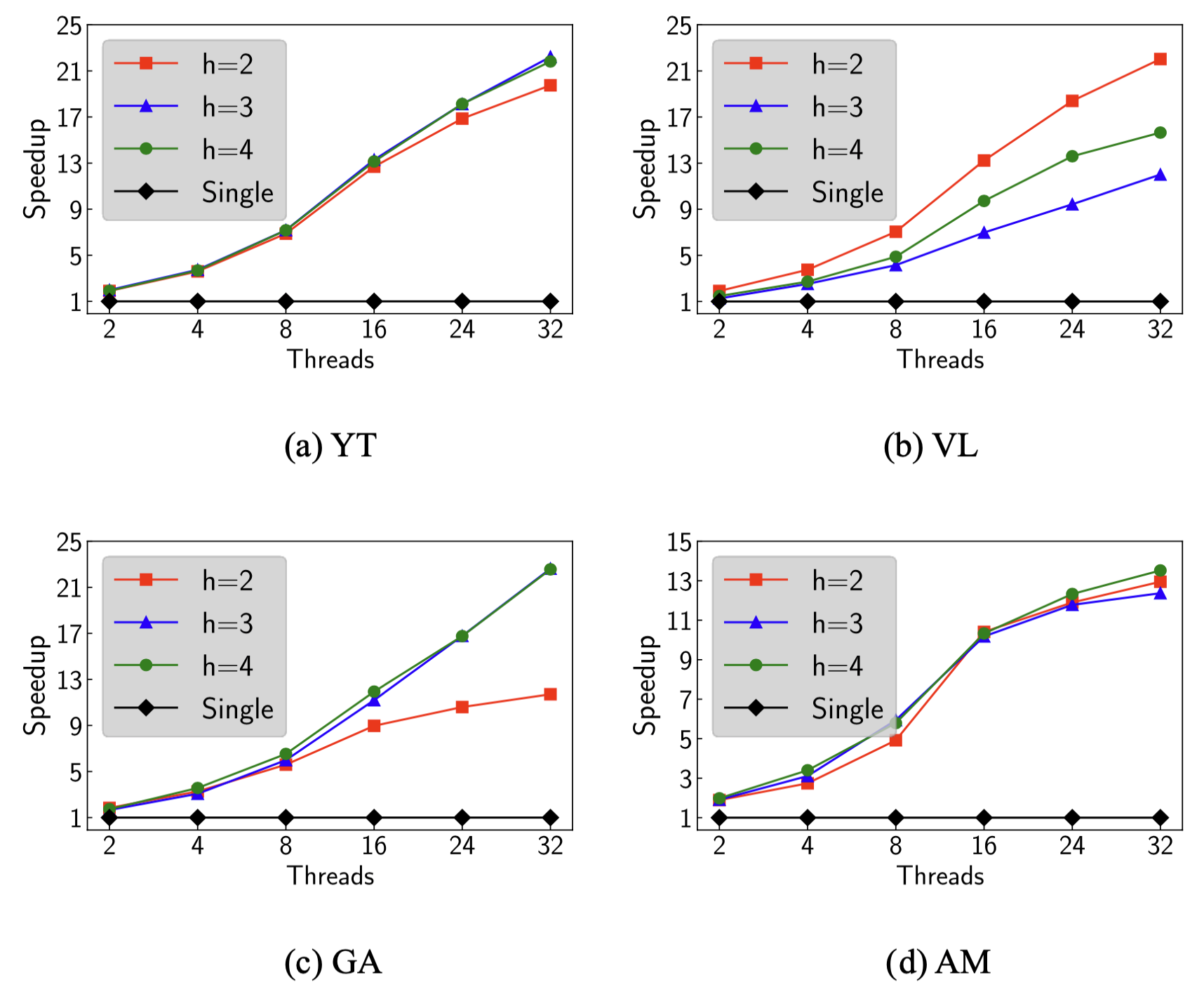}
    \caption{Evaluation of parallel framework}
    \label{fig:impact_para}
\end{figure}

\begin{figure}[t]
    \centering
    \includegraphics[width=0.9\linewidth]{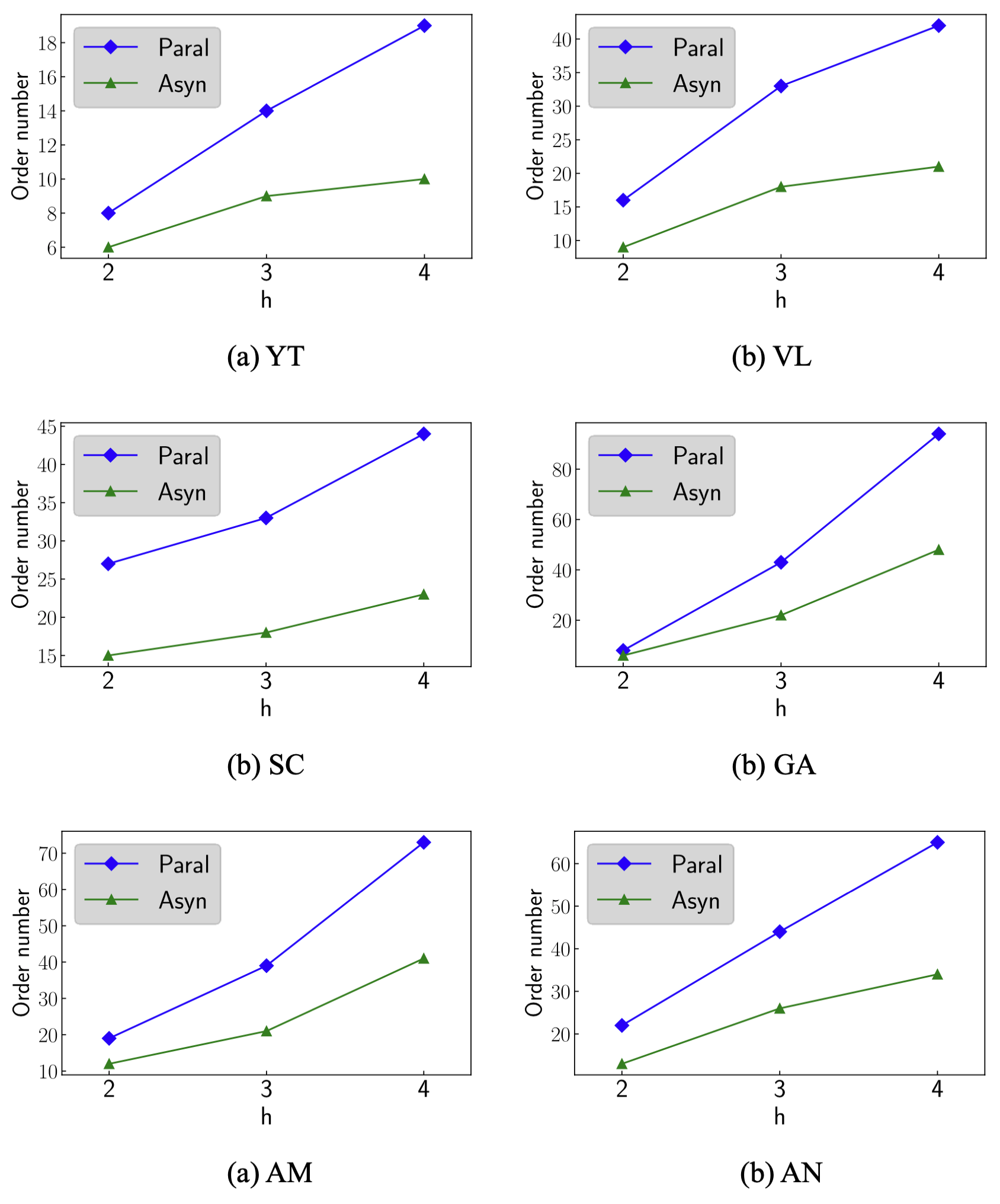}
    \caption{Evaluation of asynchronous strategy}
    \label{fig:order}
\end{figure}

\myparagraph{Evaluation of asynchronous strategy} To evaluate the impact of asynchronous strategy, we conduct the experiments on all the datasets by varing the value of $h$, and report the order number (i.e., iteration number) of Paral and Asyn. The results are shown in Figure~\ref{fig:order}. As we can see, Asyn can achieve better performance than Paral.  The Asyn algorithm can
reduce it up by nearly half. This result reflects the
effectiveness of asynchronous optimization strategy.

\section{Related work}

As one of the fundamental problems in graph analysis, cohesive subgraph discovery has been extensively studied in recent years. 
Various cohesive subgraph models have been proposed, such as $k$-core~\cite{sun2020stable},  $k$-peak~\cite{tan2023higher}, and clique~\cite{sun2023clique}. 
Among these, the $k$-truss model attracts significant interest in the literature, as it achieves a good balance between structural constraints and computational overhead. 
$k$-core requires that each node has at least $k$-neighbors.
$k$-truss further enhances the cohesion within the subgraphs identified by the $k$-core model, by requiring that the edges in a $k$-truss involved in at least $k-2$ triangles.
The cohesive constraints in $k$-truss is stronger than the $k$-core model, while weaker than the clique model, (i.e. complete subgraphs) allowing for more efficient computations.
$k$-truss can find many applications in network analysis, such as hierarchical structure analysis~\cite{cohen2008trusses}, graph visualization~\cite{huang2016truss}, and community detection~\cite{chen2014distributed}.


Nonetheless, $k$-truss lacks the ability to uncover the more fine-grained structures hidden in graphs. 
The importance of higher-order structures is well-established to identify useful features of complex networks, that is, considering the interactions between $h$-hop neighbors. 
Recently, several studies have integrated $h$-hop neighbors into cohesive subgraph models.
For example, the higher-order $k$-core was first proposed in~\cite{batagelj2011fast}, known as the $(k, h)$-core. 
Subsequent work in~\cite{bonchi2019distance} proposed an advanced algorithm to speed up the $(k, h)$-core computation. 
\cite{tan2023higher} extended the $k$-peak model to higher-order scenario and proposed the $(k,h)$-peak model. 
Higher-order $k$-truss model was first proposed in~\cite{chen2021higher}, where the authors designed efficient algorithms for $(k,h)$-truss decomposition.
However, none of the existing studies considers efficient $(k,h)$-truss decomposition using parallel method.

\section{Conclusion}

$k$-truss decomposition is widely studied in the literature to capture the properties of networks. However, the conventional $k$-truss neglects the higher-order information and limits its applications. In this paper, we investigate the problem of higher-order truss decomposition and propose the first parallel solution. In addition, several optimization techniques are proposed to accelerate the processing. Finally, we conduct experiments on real-world datasets to verify the advantage of proposed methods.

\begin{acks}
The work was supported by AEGiS (888/008/276).
\end{acks}


\bibliographystyle{ACM-Reference-Format}
\bibliography{main}

\end{document}